\documentclass{article} 

\usepackage{graphicx}      
\usepackage{cite}
\usepackage{graphicx}
\usepackage{algorithmic}
\usepackage{algorithm}
\usepackage{amssymb}  
\usepackage{amsmath} 
\usepackage{amsthm}
\usepackage{mathrsfs}
\usepackage{xspace}
\usepackage{hyperref}
\usepackage{booktabs}
\usepackage{multicol}

\newtheoremstyle{ex}
  { }
  { }
  { }
  { }
  {\bfseries}
  { }
  { }
  {\thmname{#1}\thmnumber{ #2}:\thmnote{ #3}}
\theoremstyle{ex}

\theoremstyle{remark}

\theoremstyle{plain}

\newtheorem{proposition}{Proposition}

\newcommand{\Sec}[1]{Section~\ref{sec:#1}}
\newcommand{\Fig}[1]{Figure~\ref{fig:#1}}
\newcommand{\Prop}[1]{Proposition~\ref{prop:#1}}
\newcommand{\Alg}[1]{Algorithm~\ref{alg:#1}}

\newcommand{\matlab}{MATLAB\xspace}
\newcommand{\eg}{e.g.}
\newcommand{\ie}{i.e.}
\newcommand{\NB}{N.B.\@\xspace}

\newcommand\ssm{SSM\@\xspace}

\newcommand{\+}{\mathsf{T}}                      
\DeclareMathOperator*\rank{rank}                          
\newcommand\normsum[1]{\sum_#1}    
\newcommand{\Ordo}{O}                            
\newcommand{\range}[2]{#1, \, \dots, \, #2}      
\newcommand{\crange}[2]{\{#1, \, \dots, \, #2\}} 
\newcommand{\setX}{\ensuremath{\mathsf{X}}}                      
\newcommand\parameter{\theta}
\newcommand\T{T}
\newcommand\Np{N}
\newcommand\Mp{M}

\newcommand\x{\mathbf{x}}
\renewcommand\a{\mathbf{a}}

\newcommand\pg{PG\@\xspace}
\newcommand\pgbs{PG-BS\@\xspace}
\newcommand\pgas{PG-AS\@\xspace}
\newcommand\ffbsi{FF-BS\@\xspace}
\newcommand\parspace{\Theta}

\newcommand\N{\mathcal{N}}
\newcommand\eye[1]{I_{#1}}
\newcommand\diag{\text{diag}}

\newcommand\KLD{D_\mathsf{KLD}}
\newcommand\TV{D_\mathsf{TV}}

\newcommand\dt{T_s}

\newlength\Papproxplotheight
\newlength\boxplotheight
\title{Ancestor Sampling for Particle Gibbs}

\author{
  Fredrik Lindsten\\
  Div. of Automatic Control\\
  Link{\"o}ping University\\
  \texttt{lindsten@isy.liu.se} \\
  \and
  Michael I. Jordan\\
  Dept. of EECS and Statistics\\
  University of California, Berkeley\\
  \texttt{jordan@cs.berkeley.edu}
  \and
  Thomas B. Sch\"on \\
  Div. of Automatic Control\\
  Link{\"o}ping University\\
  \texttt{schon@isy.liu.se} \\
}

\begin{document}
\maketitle

\begin{abstract}
  We present a novel method in the family of particle MCMC methods
  that we refer to as \emph{particle Gibbs with ancestor sampling} (\pgas).
  Similarly to the existing \emph{\pg with backward simulation} (\pgbs) procedure,
  we use backward sampling to (considerably) improve the mixing of the \pg kernel.
  Instead of using separate forward and backward sweeps as in \pgbs, however,
  we achieve the same effect in a single forward sweep.  We apply the \pgas
  framework to the challenging class of non-Markovian state-space models.
  We develop a truncation strategy of these models that is applicable in
  principle to any backward-simulation-based method, but which is particularly
  well suited to the \pgas framework.  In particular, as we show in a simulation 
  study, \pgas can yield an order-of-magnitude improved accuracy relative to 
  \pgbs due to its robustness to the truncation error.  Several application examples 
  are discussed, including Rao-Blackwellized particle smoothing and inference 
  in degenerate state-space models. This report is a slightly extended version of the paper \cite{LindstenJS:2012}.
\end{abstract}

\section{Introduction}\label{sec:intro}%
State-space models (\ssm{s}) are widely used to model time series and dynamical systems.
The strong assumptions of linearity and Gaussianity that were originally invoked in
state-space inference have been weakened by two decades of research on
sequential Monte Carlo (SMC) and Markov chain Monte Carlo (MCMC).  These Monte
Carlo methods have not, however, led to substantial weakening of a further strong
assumption, that of Markovianity.  It remains a major challenge to develop inference 
algorithms for non-Markovian \ssm{s}:
\begin{align}
  \label{eq:intro_nonmarkov}
  x_{t+1} &\sim f(x_{t+1} \mid \parameter, x_{1:t}), &
  y_t &\sim g(y_t \mid \parameter, x_{1:t}),
\end{align}
where $\parameter \in \parspace$ is a static parameter with prior density $p(\parameter)$,
$x_t$ is the latent state and $y_t$ is the observation at time $t$, respectively.
Models of this form arise in many different application scenarios, either from 
direct modeling or via a transformation or marginalization of a larger model.  
We provide several examples in \Sec{examples}.

To tackle the challenging problem of inference for non-Markovian \ssm{s},
we work within the framework of particle MCMC (PMCMC), a family of inferential
methods introduced in \cite{AndrieuDH:2010}.  The basic idea in PMCMC is 
to use SMC to construct a proposal kernel for an MCMC sampler.  Assume 
that we observe a sequence of measurements $y_{1:\T}$. We are interested 
in finding the density $p(x_{1:\T}, \parameter \mid y_{1:\T})$, \ie, the 
joint posterior density of the state sequence and the parameter.
In an idealized Gibbs sampler we would target this density by sampling as follows:
\emph{(i)} Draw $\theta^\star \mid x_{1:\T} \sim p(\theta \mid x_{1:\T}, y_{1:\T})$;
\emph{(ii)} Draw $x_{1:\T}^\star \mid \theta^\star \sim p(x_{1:\T} \mid \theta^\star, y_{1:\T})$.
The first step of this procedure can be carried out exactly if conjugate
priors are used.  For non-conjugate models, one option is to replace Step \emph{(i)}
with a Metropolis-Hastings step.  However, Step \emph{(ii)}---sampling from the 
joint smoothing density $p(x_{1:\T} \mid \parameter, y_{1:\T})$---is in 
most cases very difficult. In PMCMC, this is addressed by instead sampling 
a particle trajectory $x_{1:\T}^\star$ based on an SMC approximation of the
joint smoothing density. More precisely, we run an SMC sampler targeting 
$p(x_{1:\T} \mid \parameter^\star, y_{1:\T})$.  We then sample one of the 
particles at the final time $\T$, according to their importance weights, 
and trace the ancestral lineage of this particle to obtain the trajectory 
$x_{1:\T}^\star$.  This overall procedure is referred to as \emph{particle Gibbs} (\pg).

The flexibility provided by the use of SMC as a proposal mechanism for 
MCMC seems promising for tackling inference in non-Markovian models.  
To exploit this flexibility we must address a drawback of \pg in 
the high-dimensional setting, which is that the mixing of the 
\pg kernel can be very poor when there is path degeneracy in the SMC 
sampler~\cite{WhiteleyAD:2010,LindstenS:2012}.  This problem has been
addressed in the generic setting of \ssm{s} by adding a backward simulation 
step to the \pg sampler, yielding a method denoted \emph{\pg with backward 
simulation} (\pgbs).  It has been found that this considerably improves 
mixing, making the method much more robust to a small number of particles 
as well as larger data records \cite{WhiteleyAD:2010,LindstenS:2012}.

Unfortunately, however, the application of backward simulation is problematic 
for non-Markovian models.  The reason is that we need to consider full state 
trajectories during the backward simulation pass, leading to $\Ordo(\T^2)$ 
computational complexity (see \Sec{nonmarkov} for details).  To address 
this issue, we develop a novel PMCMC method which we refer to as 
\emph{particle Gibbs with ancestor sampling} (\pgas) that achieves
the effect of backward sampling without an explicit backward pass.  As part
of our development, we also develop a truncation method geared to non-Markovian 
models.  This method is a generic method that is also applicable to \pgbs, 
but, as we show in a simulation study in \Sec{numerical}, the effect of the 
truncation error is much less severe for \pgas than for \pgbs.  Indeed, we 
obtain up to an order of magnitude increase in accuracy in using \pgas when 
compared to \pgbs in this study.

Since we assume that it is straightforward to sample the parameter 
$\parameter$ of the idealized Gibbs sampler, we will not explicitly 
include sampling of $\parameter$ in the subsequent sections to simplify
our presentation.

This report is a slightly extended version of the paper \cite{LindstenJS:2012}.

\section{Sequential Monte Carlo}\label{sec:smc}%
We first review the standard auxiliary SMC sampler, see \eg~\cite{DoucetJ:2011,PittS:1999}.
Let $\gamma_{t}(x_{1:t})$ for $t = \range{1}{\T}$ be a sequence of unnormalized densities on $\setX^t$,
which we assume can be evaluated pointwise in linear time. Let $\bar\gamma_t(x_{1:t})$ be the corresponding
normalized probability densities. For an \ssm we would typically have $\bar\gamma_{t}(x_{1:t}) = p(x_{1:t} \mid y_{1:t})$
and $\gamma_t(x_{1:t}) = p(x_{1:t}, y_{1:t})$.
Assume that $\{x_{1:t-1}^m, w_{t-1}^m\}_{m=1}^\Np$ is a weighted particle system targeting $\bar\gamma_{t-1}(x_{1:t-1})$.
This particle system is propagated to time $t$ by sampling independently from a proposal kernel,
\begin{align}
  M_t(a_t, x_t) = \frac{w_{t-1}^{a_t} \nu_{t-1}^{a_t}}{\normsum{l} w_{t-1}^l \nu_{t-1}^{l}} R_t(x_t \mid x_{1:t-1}^{a_t}).
\end{align}
In this formulation, the resampling step is implicit and corresponds to sampling the ancestor indices $a_t$. Note that
$a_t^m$ is the index of the ancestor particle of $x_t^m$. When we write $x_{1:t}^m$ we refer to the ancestral path of $x_t^m$.
The factors $\nu_{t}^m = \nu_{t}(x_{1:t}^{m})$, known as adjustment multiplier weights,
are used in the auxiliary SMC sampler to increase the probability of
sampling ancestors that better can describe the current observation \cite{PittS:1999}.
The particles are then weighted according to $w_t^m = W_t(x_{1:t}^m)$, where the weight function is given by
\begin{align}
  \label{eq:smc_weightfunction}
  W_t(x_{1:t}) = \frac{\gamma_t(x_{1:t})}{ \gamma_{t-1}(x_{1:t-1}) \nu_{t-1}(x_{1:t-1}) R_t(x_t \mid x_{1:t-1})},
\end{align}
for $t \geq 2$. The procedure is initiated by sampling from a proposal density $x_1^m \sim R_1(x_1)$ and assigning importance weights
$w_1^m = W_1(x_1^m)$ with $W_1(x_1) = \gamma_1(x_1)/R_1(x_1)$. In PMCMC it is instructive to view this sampling procedure
as a way of generating a single sample from the density
\begin{align}
  \psi(\x_{1:\T}, \a_{2:\T}) \triangleq \prod_{m = 1}^\Np R_1(x_1^m) \prod_{t = 2}^{\T} \prod_{m = 1}^\Np M_t(a_t^m, x_t^m)
\end{align}
on the space $\setX^{\Np\T} \times \crange{1}{N}^{\Np(\T-1)}$. Here we have introduced the boldface notation
$\x_t = \crange{x_t^1}{x_t^\Np}$ and similarly for the ancestor indices.

\section{Particle Gibbs with ancestor sampling}\label{sec:pgas}
PMCMC methods is a class of MCMC samplers in which SMC is used to construct proposal kernels~\cite{AndrieuDH:2010}.
The validity of these methods can be assessed by viewing them as MCMC samplers on an extended state space in which all the random variables generated
by the SMC sampler are seen as auxiliary variables. The target density on this extended space is given by
\begin{align}
  \label{eq:pmcmc_phidef}
  \phi(\x_{1:\T}, \a_{2:\T}, k) \triangleq \frac{\bar\gamma_{\T}(x_{1:\T}^k)}{N^\T} \frac{  \psi(\x_{1:\T}, \a_{2:\T}) }{R_1(x_1^{b_1}) \prod_{t = 2}^{\T} M_t(a_t^{b_t}, x_t^{b_t})}.
\end{align}
By construction, this density admits $\bar\gamma_{\T}(x_{1:\T}^k)$ as a marginal,
and can thus be used as a surrogate for the original target density $\bar\gamma_{\T}$  \cite{AndrieuDH:2010}.
Here $k$ is a variable indexing one of the particles at the final time point and $b_{1:\T}$ corresponds to the ancestral path of this particle:
$x_{1:\T}^k = x_{1:\T}^{b_{1:\T}} =  \crange{x_1^{b_1}}{x_{\T}^{b_\T}}$. These indices are given recursively from
the ancestor indices by $b_\T = k$ and $b_{t} = a_{t+1}^{b_{t+1}}$.
The \pg sampler  \cite{AndrieuDH:2010} is a Gibbs sampler targeting $\phi$ using the following sweep
(note that $b_{1:\T} = \{a_{2:\T}^{b_{2:\T}}, b_\T\}$),
\begin{enumerate}
\item Draw $\x_{1:\T}^{\star,-b_{1:\T}}, \a_{2:\T}^{\star, -b_{2:\T}} \sim \phi(\x_{1:\T}^{-b_{1:\T}}, \a_{2:\T}^{-b_{2:\T}} \mid x_{1:\T}^{b_{1:\T}}, b_{1:\T})$.
\item Draw $k^\star \sim \phi(k \mid \x_{1:\T}^{\star,-b_{1:\T}}, \a_{2:\T}^{\star,-b_{2:\T}}, x_{1:\T}^{b_{1:\T}}, a_{2:\T}^{b_{2:\T}})$.
\end{enumerate}
Here we have introduced the notation $\x_t^{-m} = \{\range{x_t^1}{x_t^{m-1},\,\range{x_t^{m+1}}{x_t^\Np}} \}$,
$\x_{1:\T}^{-b_{1:\T}} = \crange{\x_1^{-b_1}}{\x_\T^{-b_\T}}$ and similarly for the ancestor indices.
In \cite{AndrieuDH:2010}, a sequential procedure for sampling from the conditional density appearing in Step 1 is given.
This method is known as \emph{conditional SMC} (CSMC). It takes the form of an SMC sampler
in which we condition on the event that a prespecified path $x_{1:\T}^{b_{1:\T}} = x_{1:\T}^\prime$, with
indices $b_{1:\T}$, is maintained throughout the sampler (see \Alg{csmc_as} for a related procedure).
Furthermore, the conditional distribution appearing in Step 2 of the \pg sampler is shown to be proportional to $w_\T^k$,
and it can thus straightforwardly be sampled from.

Note that we never sample new values for the variables $\{x_{1:\T}^{b_{1:\T}}, b_{1:\T-1}\}$ in this sweep. Hence, the \pg sampler is an ``incomplete'' Gibbs sampler, since it
does not loop over all the variables of the model. It still holds that the \pg sampler is ergodic, which intuitively can be explained by the fact
that the collection of variables that is left out is chosen randomly at each iteration. However, it has been observed that the \pg sampler can have very poor mixing,
especially when $\Np$ is small and/or $\T$ is large \cite{WhiteleyAD:2010,LindstenS:2012}. The reason for this poor mixing is that the
SMC path degeneracy causes the collections of variables that are left out at any two consecutive iterations to be strongly dependent.

We now turn to our new procedure, \pgas, which aims to address this fundamental issue.
Our idea is to sample new values for the \emph{ancestor indices} $b_{1:\T-1}$ as part of the CSMC procedure\footnote{Ideally, we would like to include the variables $x_{1:\T}^{b_{1:\T}}$ as well, but this is in general not possible since it would be similar to sampling from the original target density (which we assume is infeasible).}. By adding these variables to the Gibbs sweep, we can considerably improve the mixing of the \pg kernel.
The CSMC method is a sequential procedure to sample from $\phi(\x_{1:\T}^{-b_{1:\T}}, \a_{2:\T}^{-b_{2:\T}} \mid x_{1:\T}^{b_{1:\T}}, b_{1:\T})$
by sampling according to
$\{ \x_t^{\star,-b_t}, \a_t^{\star, -b_t} \} \sim \phi(\x_t^{-b_t}, \a_t^{-b_t} \mid \x_{1:t-1}^{\star,-b_{1:t-1}}, \a_{2:t-1}^{\star,-b_{2:t-1}}, x_{1:\T}^{b_{1:\T}}, b_{1:\T})$, for $t = \range{1}{\T}$. After having sampled these variables at time $t$, we add a step in which
we generate a new value for $b_{t-1} (= a_t^{b_t})$, resulting in the following sweep:
\begin{enumerate}
\item[$1^{\prime}$.] \textit{(CSMC with ancestor sampling)} For $t = \range{1}{\T}$, draw
  \begin{align*}
    \x_t^{\star,-b_t}, \a_t^{\star, -b_t} \sim {}& \phi(\x_t^{-b_t}, \a_t^{-b_t} \mid \x_{1:t-1}^{\star,-b_{1:t-1}}, \a_{2:t-1}^{\star}, x_{1:\T}^{b_{1:\T}}, b_{t-1:\T}), \\
    (a_{t}^{\star,b_{t}} =)~b^\star_{t-1} \sim {}&\phi(b_{t-1} \mid \x_{1:t-1}^{\star,-b_{1:t-1}}, \a_{2:t-1}^{\star}, x_{1:\T}^{b_{1:\T}}, b_{t:\T}).
  \end{align*}
\item[$2^{\prime}$.] Draw $(k^\star =)~b_\T^\star \sim {}\phi(b_\T \mid \x_{1:\T}^{\star,-b_{1:\T}}, \a_{2:\T}^{\star}, x_{1:\T}^{b_{1:\T}})$.

\end{enumerate}
It can be verified that this corresponds to a partially collapsed Gibbs sampler \cite{DykP:2008} and will thus leave $\phi$ invariant.
To determine the conditional densities from which the ancestor indices are drawn, consider the following factorization, following directly from \eqref{eq:smc_weightfunction},
\begin{align}
  \nonumber
  \gamma_t(x_{1:t}) &= W_t(x_{1:t}) \nu_{t-1}(x_{1:t-1}) R_t(x_t \mid x_{1:t-1}) \gamma_{t-1}(x_{1:t-1}) \\
  \nonumber
  \Rightarrow \gamma_t(x_{1:t}^{b_t}) &= w_t^{b_t} \frac{\sum_l w_{t-1}^l \nu_{t-1}^l }{ w_{t-1}^{b_{t-1}}} \frac{ w_{t-1}^{b_{t-1}} \nu_{t-1}^{b_{t-1}} }{\sum_l w_{t-1}^l \nu_{t-1}^l }
  R_t(x_t^{b_t} \mid x_{1:t-1}^{b_{t-1}}) \gamma_{t-1}(x_{1:t-1}^{b_{t-1}}) \\
  \label{eq:target_factorization}
  =\dots &= w_t^{b_t} \left( \prod_{s = 1}^{t-1} \sum_l w_s^l \nu_s^l \right) R_1(x_1^{b_1}) \prod_{s = 2}^t M_t(a_s^{b_s}, x_s^{b_s}).
\end{align}
Furthermore, we have
\begin{align}
  \nonumber
  \phi(b_t \mid {}&\x_{1:t}, \a_{2:t}, x_{t+1:\T}^{b_{t+1:\T}}, b_{t+1:\T}) \propto  \phi(\x_{1:t}, \a_{2:t}, x_{t+1:\T}^{b_{t+1:\T}}, b_{t:\T}) \\
  &\propto \frac{ \gamma_{\T}(x_{1:\T}^k) \psi(\x_{1:t}, \a_{2:t}) }{R_1(x_1^{b_1}) \prod_{s = 2}^t M_s(a_s^{b_s}, x_s^{b_s})}
  \propto  \frac{\gamma_t(x_{1:t}^{b_t})}{\gamma_t(x_{1:t}^{b_t})} \frac{\gamma_{\T}(x_{1:\T}^k) }{R_1(x_1^{b_1}) \prod_{s = 2}^t M_s(a_s^{b_s}, x_s^{b_s})}.
\end{align}
By plugging \eqref{eq:target_factorization} into the numerator we get,
\begin{align}
  \label{eq:bwdsim_phi}
  \phi(b_t \mid {}&\x_{1:t}, \a_{2:t}, x_{t+1:\T}^{b_{t+1:\T}}, b_{t+1:\T}) \propto w_t^{b_t} \frac{ \gamma_{\T}(x_{1:\T}^k) }{ \gamma_t(x_{1:t}^{b_t}) }.
\end{align}
Hence, to sample a new ancestor index for the conditioned path at time $t+1$, we proceed as follows.
Given $x_{t+1:\T}^\prime$ ($= x_{t+1:\T}^{b_{t+1:\T}})$ we compute the backward sampling weights,
\begin{align}
  \label{eq:bwdsim_weights}
  w_{t\mid\T}^m = w_{t}^{m} \frac{ \gamma_{\T}(\{x_{1:t}^m , x_{t+1:\T}^\prime \}) }{ \gamma_{t}(x_{1:t}^{m}) },
\end{align}
for $m = \range{1}{\Np}$. We then set $b_{t} = m$ with probability proportional to $w_{t\mid\T}^m$.

It follows that the proposed CSMC with ancestor sampling (Step $1^\prime$), conditioned on $\{x_{1:\T}^\prime, b_{1:\T}\}$, can be realized as in \Alg{csmc_as}.
The difference between this algorithm and the CSMC sampler derived in \cite{AndrieuDH:2010} lies in the ancestor sampling step 2(b)
(where instead, they set $a_t^{b_t} = b_{t+1}$).
By introducing the ancestor sampling, we break the strong dependence between the generated particle trajectories and the
path on which we condition.
We call the resulting method, defined by Steps $1^{\prime}$ and $2^{\prime}$ above, 
\emph{\pg with ancestor sampling} (\pgas).

\begin{algorithm}
  \caption{CSMC with ancestor sampling, conditioned on $\{x_{1:\T}^\prime, b_{1:\T}\}$}
  \label{alg:csmc_as}
  \begin{enumerate}
  \item \textbf{Initialize ($t = 1$):}
    \begin{enumerate}
    \item Draw $x_1^m \sim R_1(x_1)$ for $m \neq b_1$ and set $x_1^{b_1} = x_1^\prime$.
    \item Set $w_1^m = W_1(x_1^m)$ for $m = \range{1}{\Np}$.
    \end{enumerate}
  \item \textbf{for $t = \range{2}{\T}$:}
    \begin{enumerate}
    \item Draw $\{a_t^m, x_t^m\} \sim M_t(a_t, x_t)$ for $m \neq b_t$ and set $x_t^{b_t} = x_{t}^\prime$.
    \item Draw $a_t^{b_t}$ with $P(a_t^{b_t} = m) \propto w_{t-1 \mid \T}^m$.
    \item Set $x_{1:t}^m = \{x_{1:t-1}^{a_t^m}, x_t^m\}$ and $w_t^m = W_t(x_{1:t}^m)$ for $m = \range{1}{\Np}$.
    \end{enumerate}
  \end{enumerate}
\end{algorithm}

The idea of including the variables $b_{1:\T-1}$ in the \pg sampler 
has previously been suggested by Whiteley \cite{Whiteley:2010} and further 
explored in \cite{WhiteleyAD:2010,LindstenS:2012}.  This previous work,
however, accomplishes this with a explicit backward simulation pass, which, 
as we discuss in the following section, is problematic for our applications 
to non-Markovian \ssm{s}.  In the \pgas sampler, instead of requiring distinct 
forward and backward sequences of Gibbs steps as in PG with backward simulation (\pgbs), we obtain a similar 
effect via a single forward sweep.

\section{Truncation for non-Markovian models}\label{sec:nonmarkov}%
We return to the problem of inference in non-Markovian \ssm{s} of the form
shown in \eqref{eq:intro_nonmarkov}.  To employ backward sampling, 
we need to evaluate the ratio
\begin{align}
  \label{eq:nonmarkov_ratio}
  \frac{ \gamma_{\T}(x_{1:\T}) }{ \gamma_t(x_{1:t}) } = \frac{p(x_{1:\T}, y_{1:\T})}{p(x_{1:t}, y_{1:t})} 
  = \prod_{s = t+1}^\T g(y_{s} \mid x_{1:s}) f(x_s \mid x_{1:s-1}).
\end{align}
In general, the computational cost of computing the backward sampling weights will thus 
be $\Ordo(\T)$. This implies that the cost of generating a full backward 
trajectory is $\Ordo(\T^2)$. It is therefore computationally prohibitive to employ
backward simulation type of particle smoothers, as well as the \pg samplers 
discussed above, for general non-Markovian models.

To make progress, we consider non-Markovian models in which there is a
decay in the influence of the past on the present, akin to that in 
Markovian models but without the strong Markovian assumption. Hence, it is possible
to obtain a useful approximation when the product in \eqref{eq:nonmarkov_ratio} 
is truncated to a smaller number of factors, say $p$. We then replace 
\eqref{eq:bwdsim_weights} with the approximation,
\begin{align}
  \label{eq:nonmarkov_truncated_weights}
   \widetilde w_{t\mid\T}^{p,m} = w_t^{m} \frac{ \gamma_{t+p}(\{x_{1:t}^m , x_{t+1:t+p}^\prime \}) }{ \gamma_t(x_{1:t}^{m}) }.
\end{align}
The following proposition formalizes our assumption.
\begin{proposition}\label{prop:truncation_kld}%
  Let $P$ and $\widetilde P_p$ be the probability distributions on $\crange{1}{\Np}$,
  defined by the backward sampling weight \eqref{eq:bwdsim_weights} and the truncated backward sampling weights \eqref{eq:nonmarkov_truncated_weights}, respectively.
  Let $h_s(k) = g(y_{t+s} \mid x_{1:t}^k, x_{t+1:t+s}^\prime) f(x_{t+s}^\prime \mid x_{1:t}^k, x_{t+1:t-s}^\prime)$ and assume that
  $\max_{k,l} \left( h_s(k) / h_s(l) -1 \right) \leq A\exp(-cs)$,
  for some constants $A$ and $c > 0$. Then, $\KLD(P \| \widetilde P_p) \leq C\exp(-cp)$ for some constant $C$, where $\KLD$ is the Kullback-Leibler divergence (KLD).
\end{proposition}
\begin{proof}
  See Appendix~\ref{app:proofs}.
\end{proof}
From \eqref{eq:nonmarkov_truncated_weights}, we see that we can compute the 
backward weights in constant time under the truncation within the \pgas framework. 
The resulting approximation can be quite useful; indeed, in our experiments we have 
seen that even $p = 1$ can lead to very accurate inferential results.  In general,
however, it will not be known a priori how to set the truncation level $p$ for any 
given problem.  To address this problem, we propose to use an adaption of the truncation 
level.  Since the approximative weights \eqref{eq:nonmarkov_truncated_weights} 
can be evaluated sequentially, the idea is to start with $p = 1$ and then increase 
$p$ until the weights have, in some sense, converged.  In particular, in our
experimental work, we have used the following simple approach.

Let $\widetilde P_p$ be the discrete probability measure defined by 
\eqref{eq:nonmarkov_truncated_weights}.  
Let $\varepsilon_p = \TV(\widetilde P_p, \widetilde P_{p-1})$
be the total variation (TV) distance between the distributions for two consecutive truncation levels. We then compute the exponentially decaying moving average
of the sequence $\varepsilon_p$, with forgetting factor $\gamma \in [0,\,1]$, and stop when this falls below some threshold $\tau \in [0,\,1]$.
This adaption scheme removes the requirement to specify $p$ directly, but instead introduces the design parameters $\gamma$ and $\tau$.
However, these parameters are much easier to reason about -- a small value for $\gamma$ gives a rapid response to changes in $\varepsilon_p$
whereas a large value gives
a more conservative stopping rule, improving the accuracy of the approximation at the cost of higher computational complexity.
A similar trade off holds for the threshold $\tau$ as well. Most importantly, we have found that the same values
for $\gamma$ and $\tau$ can be used for a wide range of models, with very different mixing properties.

To illustrate the effect of the adaption rule, and how the distribution $\widetilde P_p$ typically evolves as we increase $p$, we provide
two examples in \Fig{degen_Papprox}. These examples are taken from the simulation study provided in \Sec{numerical_degen_rndlgss}.
Note that the untruncated distribution $P$ is given for the maximal value of $p$, \ie, furthest to the right in the figures.
By using the adaptive truncation, we can stop the evaluation of the weights at a much earlier stage, and still obtain an accurate
approximation of $P$.

\begin{figure}[ptb]
  \centering
  \includegraphics[height = \Papproxplotheight]{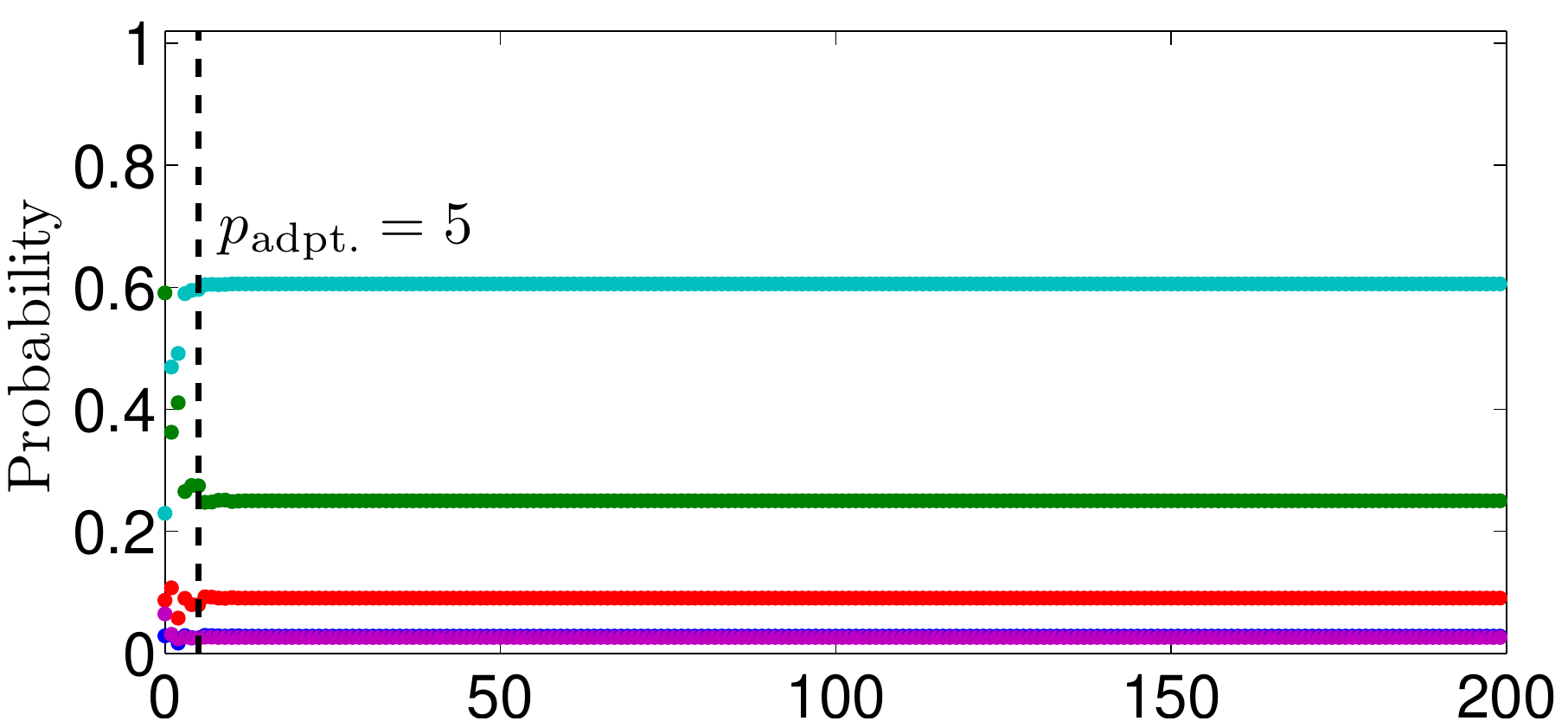}
  \includegraphics[height = \Papproxplotheight]{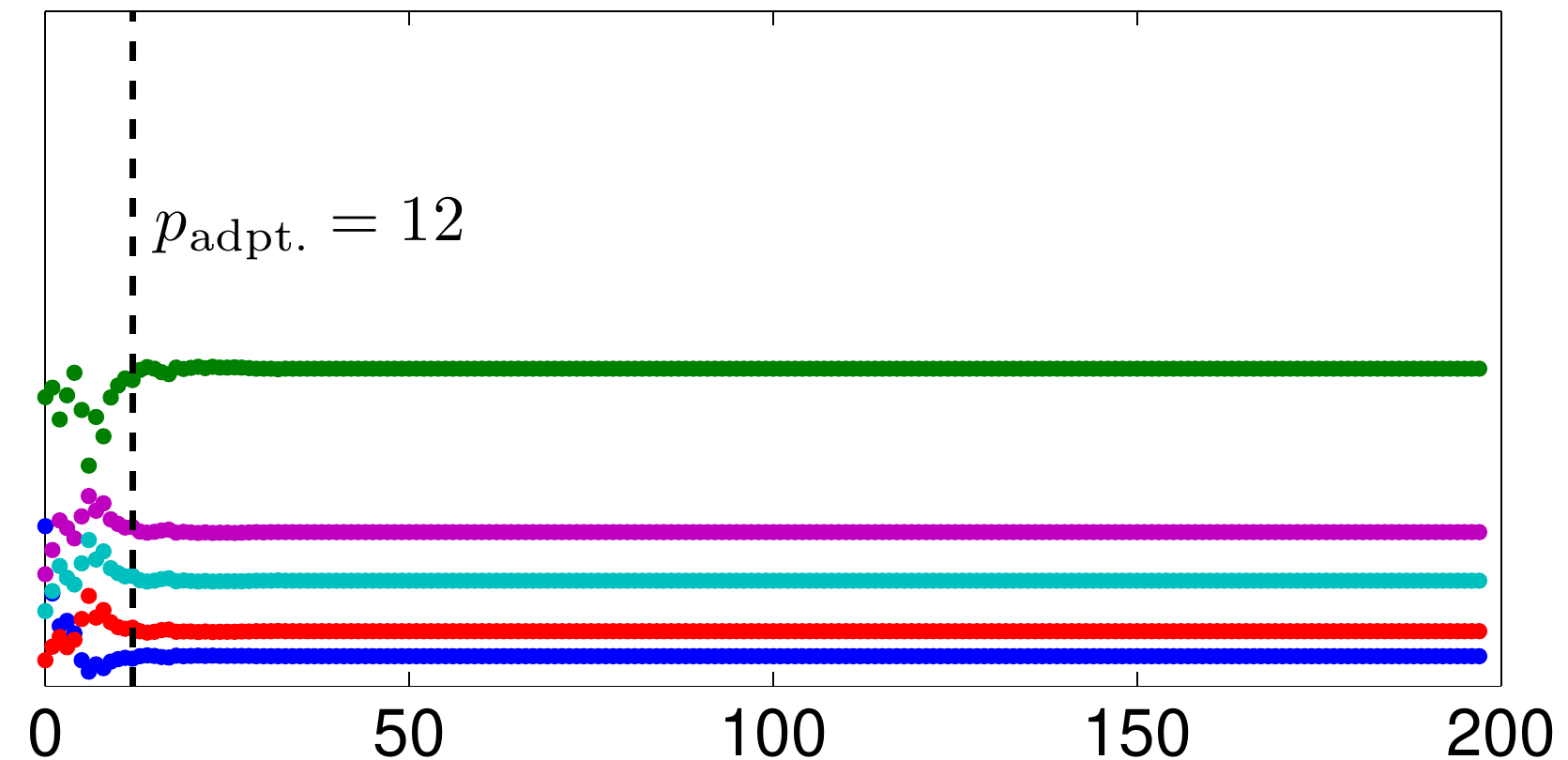}
  \caption{Probability under $\widetilde P_p$ as a function of the truncation level $p$
    for two different systems; one 5 dimensional (left) and one 20 dimensional (right).
    The $\Np = 5$ dotted lines correspond to $\widetilde P_p(m)$ for $m \in \crange{1}{\Np}$, respectively (\NB two of the lines overlap in the left figure).
    The dashed vertical lines show the
    value of the truncation level $p_{\mathrm{adpt.}}$, resulting from the adaption scheme with $\gamma = 0.1$ and $\tau = 10^{-2}$.
    See \Sec{numerical_degen_rndlgss} for details on the experiments.}
  \label{fig:degen_Papprox}
\end{figure}

\section{Application areas}\label{sec:examples}%
In this section we present examples of problem classes involving
non-Markovian \ssm{s} for which the proposed \pgas sampler
can be applied. Numerical illustrations are provided in \Sec{numerical}.

\subsection{Rao-Blackwellized particle smoothing}\label{sec:rbps}%
One popular approach to increase the efficiency of SMC samplers for \ssm{s} is to marginalize over one component
of the state, and apply an SMC sampler in the lower-dimensional marginal space.
This leads to what is known as the Rao-Blackwellized particle filter (RBPF) \cite{ChenL:2000,DoucetGA:2000,SchonGN:2005}.
The same approach has also been applied to state smoothing \cite{SarkkaBG:2012,FongGDW:2002}, but it turns out
that Rao-Blackwellization is less straightforward in this case, since the marginal state-process will be non-Markovian.
As an example, a mixed linear/nonlinear Gaussian \ssm (see, \eg, \cite{SchonGN:2005}) with ``nonlinear state'' $x_t$
and ``conditionally linear state'' $z_t$, can be reduced to
\begin{align}
  x_t &\sim p(x_t \mid x_{1:t-1}, y_{1:t-1}), &
  y_t &\sim p(y_t \mid x_{1:t}, y_{1:t-1}).
\end{align}
These conditional densities are Gaussian and can be evaluated for any fixed marginal state trajectory $x_{1:t-1}$ by running a conditional Kalman filter
to marginalize the $z_t$-process.

In order to apply a backward-simulation-based method (\eg, a particle smoother) 
for this model, we need to evaluate the backward sampling weights \eqref{eq:bwdsim_weights}.
In a straightforward implementation, we thus need to run $\Np$ Kalman filters for $\T-t$ time steps, 
for each $t = \range{1}{\T-1}$.
The computational complexity of this calculation can be reduced by employing the 
truncation proposed in \Sec{nonmarkov}\footnote{For the specific problem of Rao-Blackwellized smoothing
in conditionally Gaussian models, a backward simulator which can be implemented in $\Ordo(\T)$ computational
complexity has recently been proposed in \cite{SarkkaBG:2012}. This is based on the idea of propagating information
backward in time as the backward samples are generated.}.

\subsection{Particle smoothing for degenerate state-space models}\label{sec:degenerate}%
Many dynamical systems are most naturally modelled as degenerate in the sense that the transition kernel of the state process does not admit any dominating measure.
For instance, consider a nonlinear system with additive noise of the form,
\begin{align}
  \label{eq:degenssm_model1}
  \xi_t &= f(\xi_{t-1}) + G\omega_{t-1}, &
  y_t &= g(\xi_t) + e_t,
\end{align}
where $G$ is a tall matrix, and consequently $\rank(G) < \dim(\xi_t)$. That is, the process noise covariance matrix is singular. SMC samplers can straightforwardly be
applied to this type of models, but it is more problematic to address the smoothing problem using particle methods. The
reason is that the backward kernel also will be degenerate and it cannot be approximated in a natural way by the forward filter particles,
as is normally done in backward-simulation-based particle smoothers.

A possible remedy for this issue is to recast the degenerate \ssm as a non-Markovian model in a lower-dimensional space.
Let $G = U \begin{bmatrix} \Sigma & 0 \end{bmatrix}^\+ V^\+$ with unitary $U$ and $V$ be a singular value decomposition of $G$ and
let,
\begin{align}
  \begin{bmatrix}
    x_t \\ z_t
  \end{bmatrix} \triangleq U^\+ \xi_t = U^\+ f(UU^\+ \xi_{t-1}) +
  \begin{bmatrix}
    \Sigma V^\+ \omega_{t-1} \\ 0
  \end{bmatrix}.
\end{align}
For simplicity we assume that $z_1$ is known.
If this is not the case, it can be included in the system state or seen as a static parameter of the model. 
Hence, the sequence $z_{1:t}$ is $\sigma(x_{1:t-1})$-measurable and we can write $z_t = z_t(x_{1:t-1})$.
With $v_{t} \triangleq \Sigma V^\+ \omega_{t}$ and by appropriate definitions of the functions $f_x$ and $h$,
the model \eqref{eq:degenssm_model1} can thus be rewritten as,
\begin{align}
  x_t &= f_x(x_{1:t-1}) + v_{t-1}, &
  y_t &= h(x_{1:t}) + e_t,
\end{align}
which is a non-degenerate, non-Markovian \ssm. By exploiting the truncation proposed in \Sec{nonmarkov}
we can thus apply \pgas to do inference in this model.
In fact, this is nothing but another application of Rao-Blackwellization as discussed in \Sec{rbps}, where the $z_t$-state is conditionally
deterministic and thus trivially marginalizable.

\subsection{Additional problem classes}
There are many more problem classes in which non-Markovian models arise and in which 
backward-simulation-based methods can be of interest.
For instance, the Dirichlet process mixture model (DPMM, see, \eg, \cite{HjortHMW:2010}) is a popular nonparametric Bayesian model for mixtures
with an unknown number of components. Using a Polya urn representation, the mixture labels are given by a non-Markovian stochastic process,
and the DPMM can thus be seen as a non-Markovian \ssm. SMC has previously been used for inference in DPMMs
\cite{MacEachernCL:1999,Fearnhead:2004}. An interesting venue for future work 
is to use the \pgas sampler for these models. 
A second example in Bayesian nonparametrics is Gaussian process (GP) regression and classification (see, \eg, \cite{RasmussenW:2006}).
The sample path of the GP can be seen as the state-process in a non-Markovian SSM.
We can thus employ PMCMC, and in particular \pgas, to address these inference problems.

An application in genetics, for which SMC has been been successfully applied, is reconstruction of phylogenetic trees \cite{Bouchard-CoteSJ:2012}.
A phylogenetic tree is a binary tree with observation at the leaf nodes. SMC is used to construct the tree in a bottom up fashion.
A similar approach has also been used for Bayesian agglomerative clustering, in which SMC is used to construct a binary clustering tree based
on Kingman's coalescent \cite{TehDR:2008}. The generative models for the trees used in \cite{Bouchard-CoteSJ:2012,TehDR:2008} are in fact
Markovian, but the observations give rise to a conditional dependence which destroys the Markov property. To employ backward simulation
to these models, we are thus faced with problems of a similar nature as those discussed in \Sec{nonmarkov}.

\section{Numerical evaluation}\label{sec:numerical}
This section contains a numerical evaluation of the proposed method.
First, we consider linear Gaussian systems, which is instructive since the exact smoothing density then is available,
\eg, by running a modified Bryson-Frazier (MBF) smoother \cite{Bierman:1973}.
Second, we apply the proposed method for joint state and parameter inference in a target tracking scenario.

\subsection{RBPS: Linear Gaussian state-space model}
As a first example, we consider Rao-Blackwellized particle smoothing (RBPS) in
a single-output 4th-order linear Gaussian \ssm.
The system has poles in $-0.65$, $-0.12$ and $0.22 \pm 0.10i$ and is excited by white Gaussian noise
with variance $0.1\eye{4}$. The scalar output $y_t$ is observed in white Gaussian noise with variance $0.1$.
We generate $\T = 100$ samples from the system and run \pgas and \pgbs, marginalizing
three out of the four states using an RBPF, \ie, $\dim(x_t) = 1$. Both methods are run for $R = 10000$ iterations using $\Np = 5$ particles.
The truncation level is set to $p=1$, leading to a coarse approximation.
The total computational complexity for each sampler is $\Ordo(R\Np\T p)$.
We discard the first 1000 iterations and then compute running means of the state trajectory $x_{1:\T}$.
From these, we then compute the running root mean squared errors (RMSEs) $\epsilon_r$ \emph{relative to the true posterior means} (computed with an MBF smoother).
Hence, if no approximation would have been made, we would expect $\epsilon_r \rightarrow 0$, so any static error can be seen as the effect
of the truncation. The results for five independent runs from both \pg samplers are shown in \Fig{rbps1_res}.
First, we note that both methods give accurate results. Still,
the error for \pgas is close to an order of magnitude less than for \pgbs. Furthermore,
it appears as if the error for \pgas would decrease further, given more iterations, suggesting that the bias
caused by the truncation is dominated by the Monte Carlo variance, even after $R = 10000$ iterations.

\begin{figure}[ptb]
  \centering
  \includegraphics[width = 0.9\columnwidth]{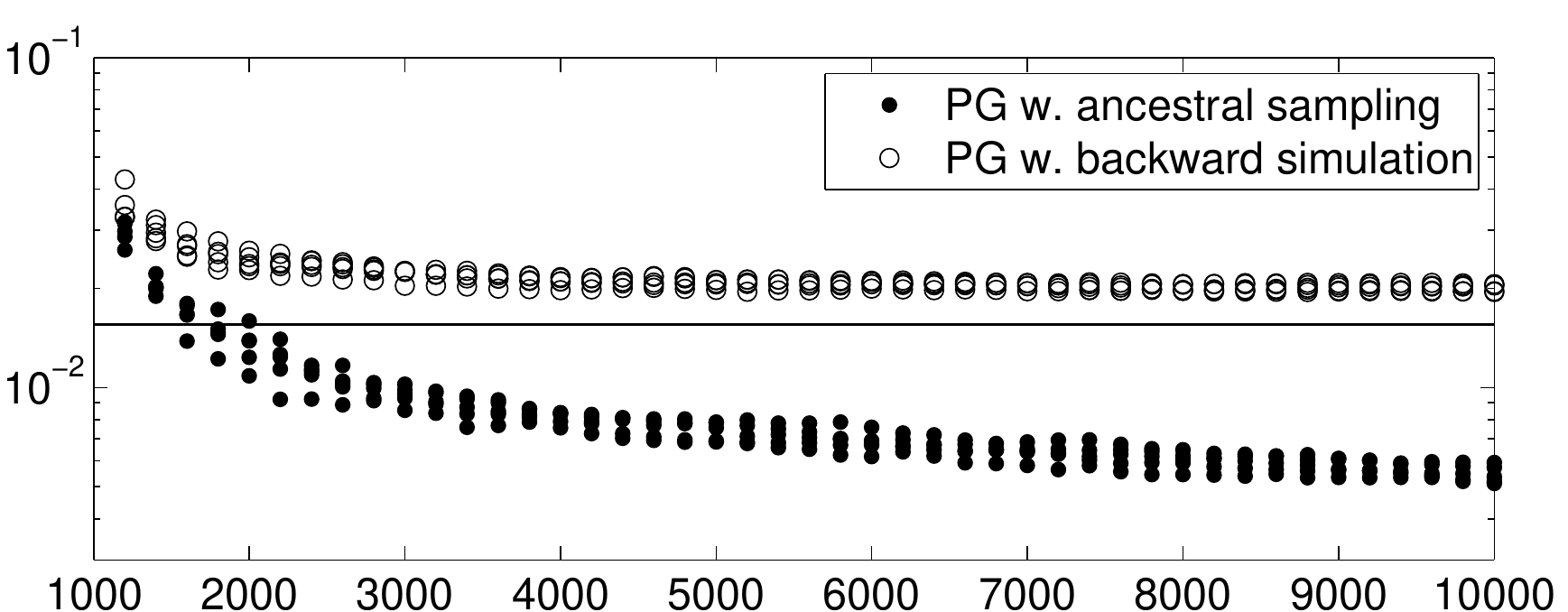}
  \caption{Rao-Blackwellized state smoothing using \pg.
    Running RMSEs for five independent runs of \pgas ($\bullet$) and \pgbs ($\circ$), respectively. The truncation level is set to $p = 1$.
    The solid line corresponds to a run of an untruncated \ffbsi.}
  \label{fig:rbps1_res}
\end{figure}

For further comparison, we also run an untruncated forward filter/backward simulator (\ffbsi) particle smoother \cite{GodsillDW:2004}, using $\Np = 5000$ forward filter particles and
$\Mp = 500$ backward trajectories (with a computational complexity of $\Ordo(\Np\Mp\T^2)$). The resulting RMSE value is shown as a solid line in \Fig{rbps1_res}.
These results suggest that PMCMC samplers, such as the \pgas, indeed can be serious competitors to more ``standard'' particle smoothers.
Even with $p=1$, \pgas outperforms \ffbsi in terms of accuracy and, due to the fact that the ancestor sampling allows us to use as few as $\Np = 5$ particles at each iteration,
at a lower computational cost.

\subsection{Random linear Gaussian systems with rank deficient process noise covariances}\label{sec:numerical_degen_rndlgss}%
To see how the \pg samplers are affected by the choice of truncation level $p$ and by the mixing properties of the system,
we evaluate them on random linear Gaussian \ssm{s} of different orders. We generate 150 random systems, using the \matlab function
\verb+drss+ from the Control Systems Toolbox, with model orders 2, 5 and 20 (50 systems for each model order).
The number of outputs are taken as 1, 2 and 4 for the different model orders, respectively.
The systems are then simulated for $\T = 200$ time steps, driven by Gaussian process noise entering only on the first state component.
Hence, the rank of the process noise covariance is 1 for all systems.
The process noise and measurement noise variances are both set to 0.1.

We run the \pgas and \pgbs samplers for 10000 iterations using $\Np = 5$ particles.
We consider different fixed truncation levels,
($p = 1$, $2$ and $3$ for 2nd order systems and $p = 1$, $5$ and $10$ for 5th and 20th order systems),
as well as an adaptive level with $\gamma = 0.1$ and $\tau = 10^{-2}$. 
Again, we compute running posterior means (discarding 1000 samples) and RMSE values relative the true posterior mean.
Box plots are shown in \Fig{degen_rndlgss}.
Since the process noise only enters on one of the state components, the mixing tends to deteriorate as we increase the model order.
\Fig{degen_Papprox} shows how the probability distributions on $\crange{1}{\Np}$ change as we increase the truncation level,
in two representative cases for a 5th and a 20th order system, respectively.
By using an adapted level, we can obtain accurate results
for systems of different dimensions, without having to change any settings between the runs.

\begin{figure}[ptb]
  \centering
  \includegraphics[height = \boxplotheight]{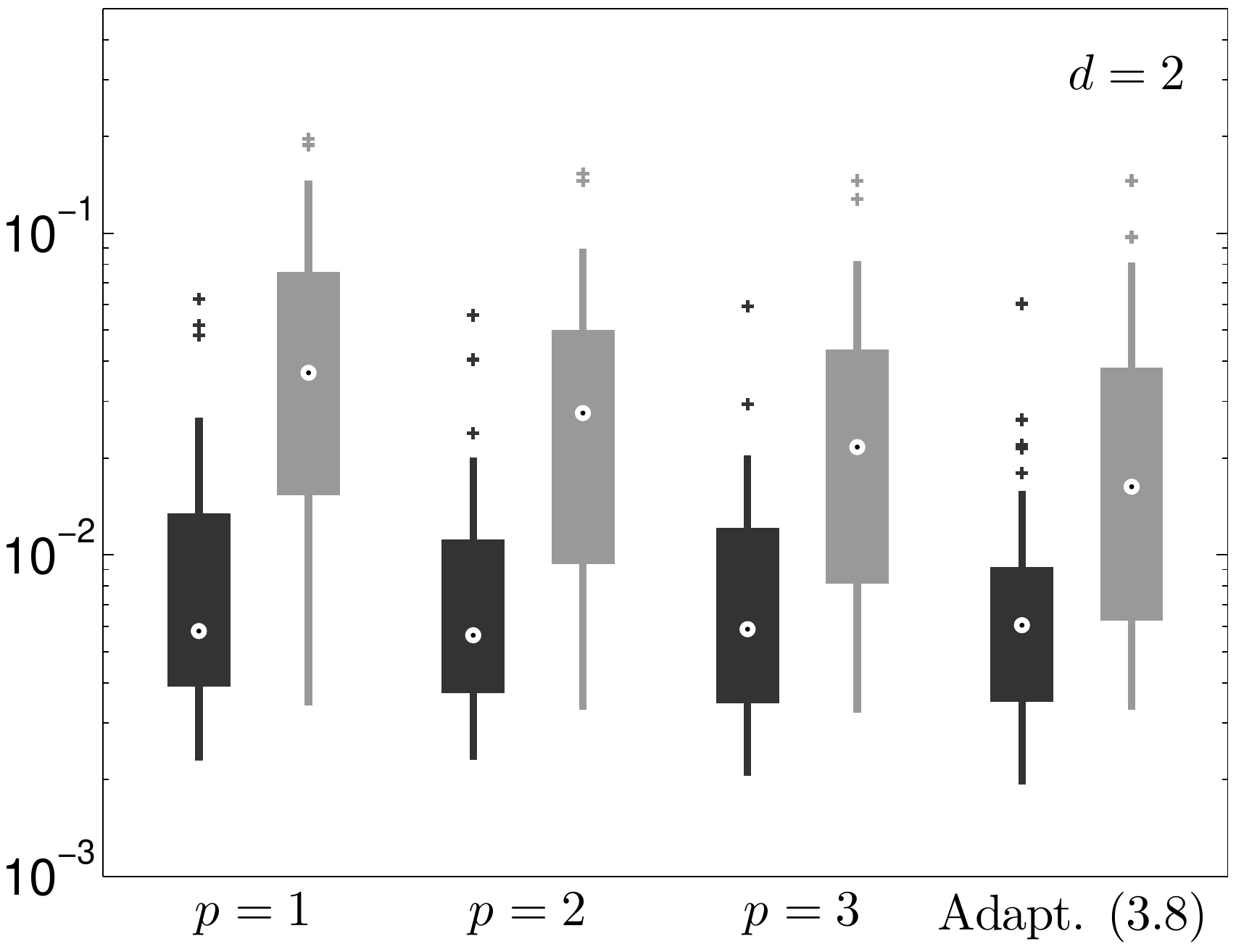}
  \includegraphics[height = \boxplotheight]{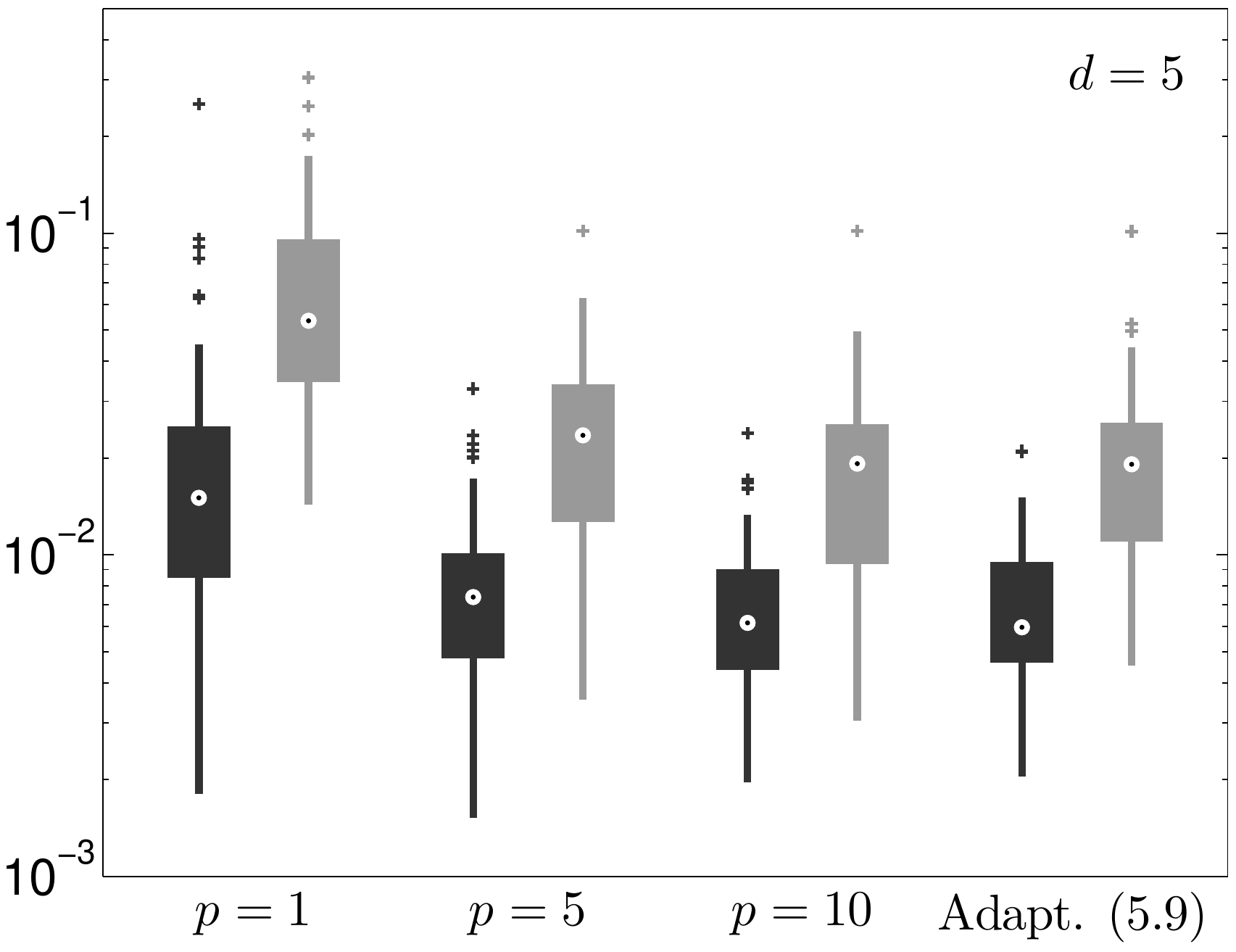}
  \includegraphics[height = \boxplotheight]{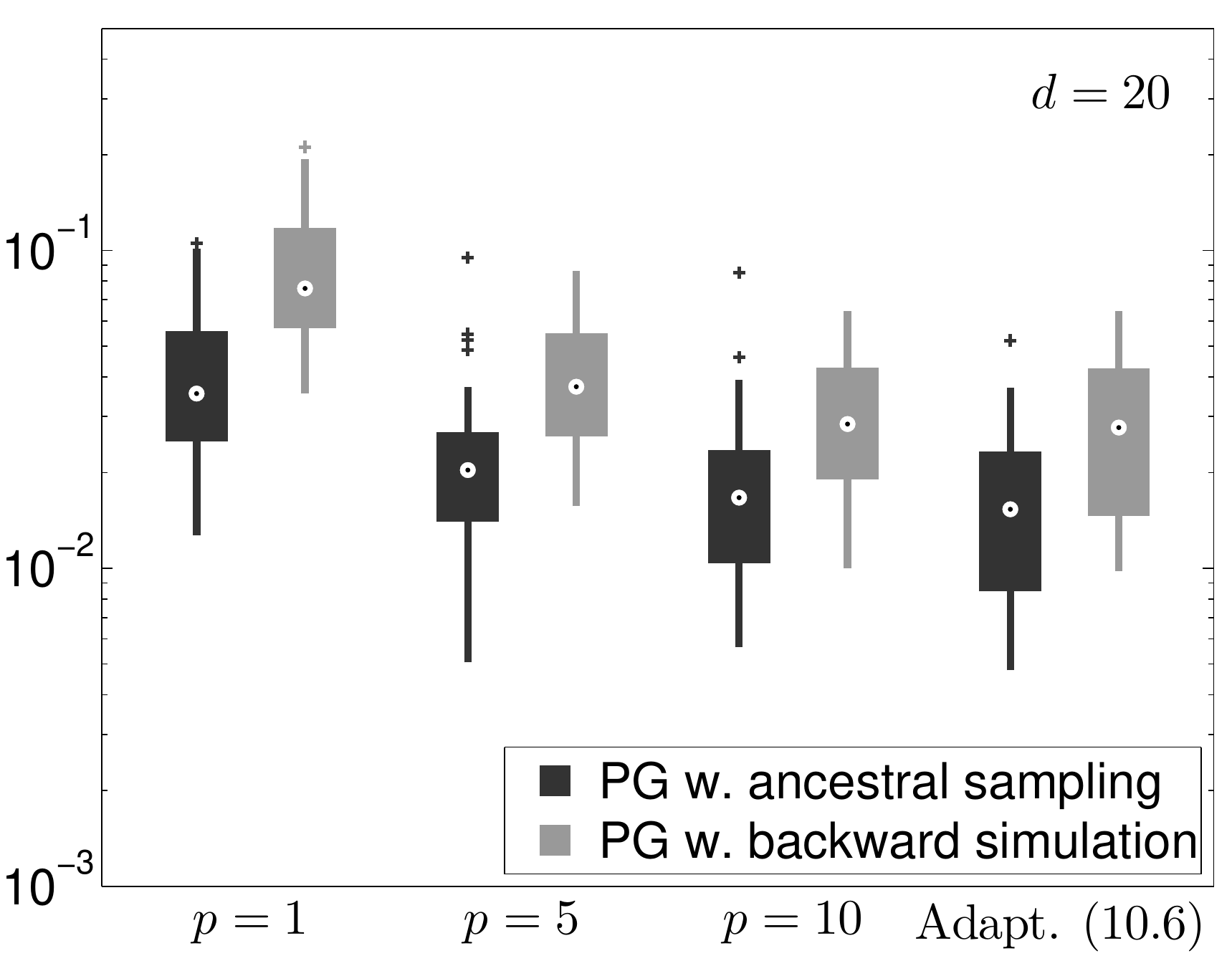}
  \caption{Box plots of the RMSE errors for \pgas (black) and \pgbs (gray), 
    for 150 random systems of different dimensions $d$
    (upper left, $d = 2$; upper right, $d = 5$; bottom, $d = 20$). Different values for the truncation level $p$ are considered. The rightmost boxes correspond
    to an adaptive threshold and the values in parentheses are the average over all systems and MCMC iterations (the same for both methods).
    The dots within the boxes show the median errors.}
  \label{fig:degen_rndlgss}
\end{figure}

\subsection{Range-bearing tracking in model with rank deficient process noise covariance}\label{sec:degen1}%
Target tracking is an area in which SMC methods have been applied with great success, see \eg
\cite{ArulampalamMGC:2002,ArulampalamRGM:2004,RisticAG:2004}.
Tracking is most commonly seen as an online filtering problem, though in certain scenarios
it might be beneficial to instead view it as a smoothing problem. For instance, if a target tracker
in a surveillance system detects some abnormal behaviour, it can be interesting to apply a smoother
to obtain refined estimates of the target's position prior to the detection.

Here, we consider smoothing in a range-bearing target tracking scenario. The system state consists
of the target's position and velocity in two dimensions, $\xi_t =
\begin{pmatrix}
  p^x_t & p^y_t & v^x_t & v^y_t
\end{pmatrix}^\+$.
We use a coordinated turn (CT) model, which is a standard model for a manoeuvring target (see \eg \cite{ArulampalamRGM:2004}),
\begin{subequations}
  \label{eq:degen1_ctmodel}
  \begin{align}
    \begin{pmatrix}
      p^x_t \\ p^y_t \\ v^x_t \\ v^y_t
    \end{pmatrix}^\+
    =
    \underbrace{
      \begin{pmatrix}
        p^x_{t-1} + \frac{\sin(\dt\Phi_{t-1} )}{\Phi_{t-1}} v^x_{t-1} - \frac{1-\cos(\dt\Phi_{t-1})}{\Phi_{t-1}} v^y_{t-1} \\
        p^y_{t-1} + \frac{1-\cos(\dt\Phi_{t-1})}{\Phi_{t-1}} v^x_{t-1} + \frac{\sin(\dt\Phi_{t-1} )}{\Phi_{t-1}} v^y_{t-1} \\
        \cos(\dt\Phi_{t-1}) v^x_{t-1} - \sin(\dt\Phi_{t-1}) v^y_{t-1} \\
        \sin(\dt\Phi_{t-1}) v^x_{t-1} + \cos(\dt\Phi_{t-1}) v^y_{t-1}
      \end{pmatrix}
    }_{= f_\parameter(\xi_{t-1})}
    +
    \underbrace{
      \begin{pmatrix}
        \frac{\dt^2}{2} & 0 \\
        0 & \frac{\dt^2}{2} \\
        \dt & 0 \\
        0 & \dt
      \end{pmatrix}}_{= G} \omega_{t-1}.
  \end{align}
  The turn rate is given by
  \begin{align}
    \Phi_t = \frac{\parameter}{\sqrt{(v^x_t)^2 + (v^y_t)^2}},
  \end{align}
\end{subequations}
which depends nonlinearly on the system state.
The parameter $\parameter$ is the manoeuvre acceleration, which we assume is fixed but unknown. This
is done to illustrate the fact that \pgas straightforwardly can be used for joint parameter and state inference,
as pointed out in \Sec{intro}. The system is assumed to be affected by a random acceleration $\omega_t \sim \N(0,Q)$ (the process noise), here with $Q = 10\eye{2}$.
This is a common assumption for many models
used in target tracking. The matrix $G$ arises from a time discretization of a continuous time model, where $\dt = 0.1$ is the sampling time.
The initial state of the system is given a Gaussian prior, $\N\left(
\begin{pmatrix}
  500 & 500 & 0 & 0
\end{pmatrix}^\+, \diag\left(
  \begin{pmatrix}
    20 & 20 & 5 & 5
  \end{pmatrix}^\+
\right) \right)$.

We assume that the range and bearing of the target can be observed, so that the measurements are given by,
\begin{align}
  y_t &=
  \begin{pmatrix}
    \sqrt{(p^x_t)^2 + (p^y_t)^2} \\
    \arctan(p^x_t / p^y_t)
  \end{pmatrix} + e_t, & e_t &\sim \N\left(0,
    \begin{pmatrix}
      50 & 0 \\ 0 & 10^{-4}
    \end{pmatrix}\right ).
\end{align}
This choice of measurement noise covariance corresponds to an accurate bearing measurement, but an uninformative range measurement.
Such a measurement could for instance arise in visual tracking, where the range is estimated based on the size of the target.

We initialize the system as $\xi_1 =
\begin{pmatrix}
  490 & 490 & 0 & 5
\end{pmatrix}^\+$ and simulate it for $\T = 200$ time steps. The
true target trajectory is shown in \Fig{degen1_position}.
Note that the process noise covariance $GQG^\+$ is singular, which implies that care needs to be taken when designing a smoothing algorithm
for this model. Here, we apply a linear state transformation, as suggested in \Sec{degenerate}, to reduce the model to a lower-dimensional
state-space. With $G = U \begin{bmatrix} \Sigma & 0 \end{bmatrix}^\+ V^\+$ we define $\begin{bmatrix} x_t^\+ & z_t^\+ \end{bmatrix}^\+ = U^\+ \xi_t$.
We then employ the \pgas sampler for joint parameter and state inference, by targeting the density $p(\parameter, z_1, x_{1:\T} \mid y_{1:\T})$.
We apply a Metropolis-Hastings step to update $\parameter$, using
a Gaussian random walk proposal with standard deviation $\sigma = 0.2$ and target density $p(\parameter | z_1, x_{1:\T}, y_{1:\T})$. The initial state of the system is unknown,
so the variable $z_1$ is seen as a part of the system state. That is, the SMC sampler targets the sequence of densities $p_\parameter(z_1, x_{1:t} \mid y_{1:t})$
for $t = \range{1}{\T}$.

\begin{figure}[ptb]
  \centering
  \includegraphics[width = 0.7\linewidth]{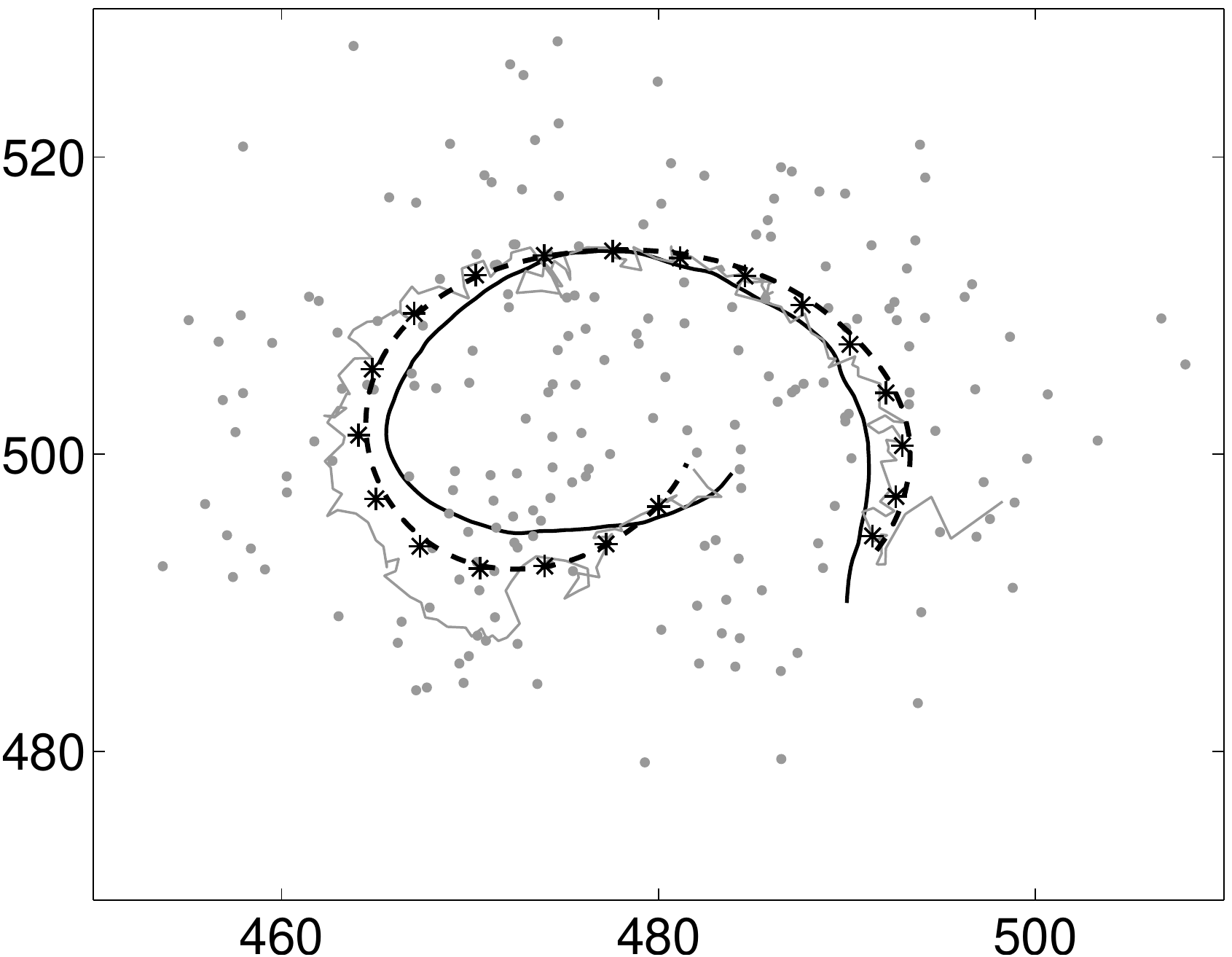}
  \caption{Target trajectory in the horizontal plane (black solid line) and smoothed posterior means for \pgas (dashed line)
    and PMMH ($*$) under parameter uncertainty. The gray line is the PF estimate and the dots show the
    range-bearing measurements, transformed to Cartesian coordinates.} 
  \label{fig:degen1_position}
\end{figure}

It is worth to point out that this SMC sampler is not more complicated to implement than a sampler targeting the original model \eqref{eq:degen1_ctmodel}.
In fact, a natural way to do the implementation is to run the sampler as if targeting $\begin{bmatrix} x_t^\+ & z_t^\+ \end{bmatrix}^\+$ jointly\footnote{For the results presented here, we used 
a standard bootstrap PF, which is very straightforward to implement.}.
The difference is that the $z_t$-particles are seen as conditional sufficient statistics for the $z_t$-state (which is conditionally deterministic),
similarly to how one propagates the sufficient statistics for the conditionally linear state in an RBPF.
The difference lies in how the backward sampling is done, where in the marginal model
we only consider the $x_{t}$-states when computing the backward weights.

The \pgas sampler was run with $\Np = 5$ particles for 50000 iterations, with the first 10000 samples discarded as burnin.
We used an adaptive truncation level with $\gamma = 0.1$ and $\tau = 10^{-2}$ (same as before), resulting in an average truncation level of $2.3$.
As a comparison, we also employ a particle marginal Metropolis-Hastings (PMMH) sampler \cite{AndrieuDH:2010}, with $\Np = 5000$ particles,
also running for 50000 MCMC iterations (discarding the first 10000 samples).
The smoothed estimates of the target trajectory are shown in \Fig{degen1_position} and the posterior density of $\parameter$ is given in \Fig{degen1_theta}.
From these results we see that the \pgas sampler provides accurate inferential results, despite the truncation of the backward weights and
without any problem specific tuning of the variables $\gamma$ and $\tau$.

\begin{figure}[ptb]
  \centering
  \includegraphics[width = 0.7\linewidth]{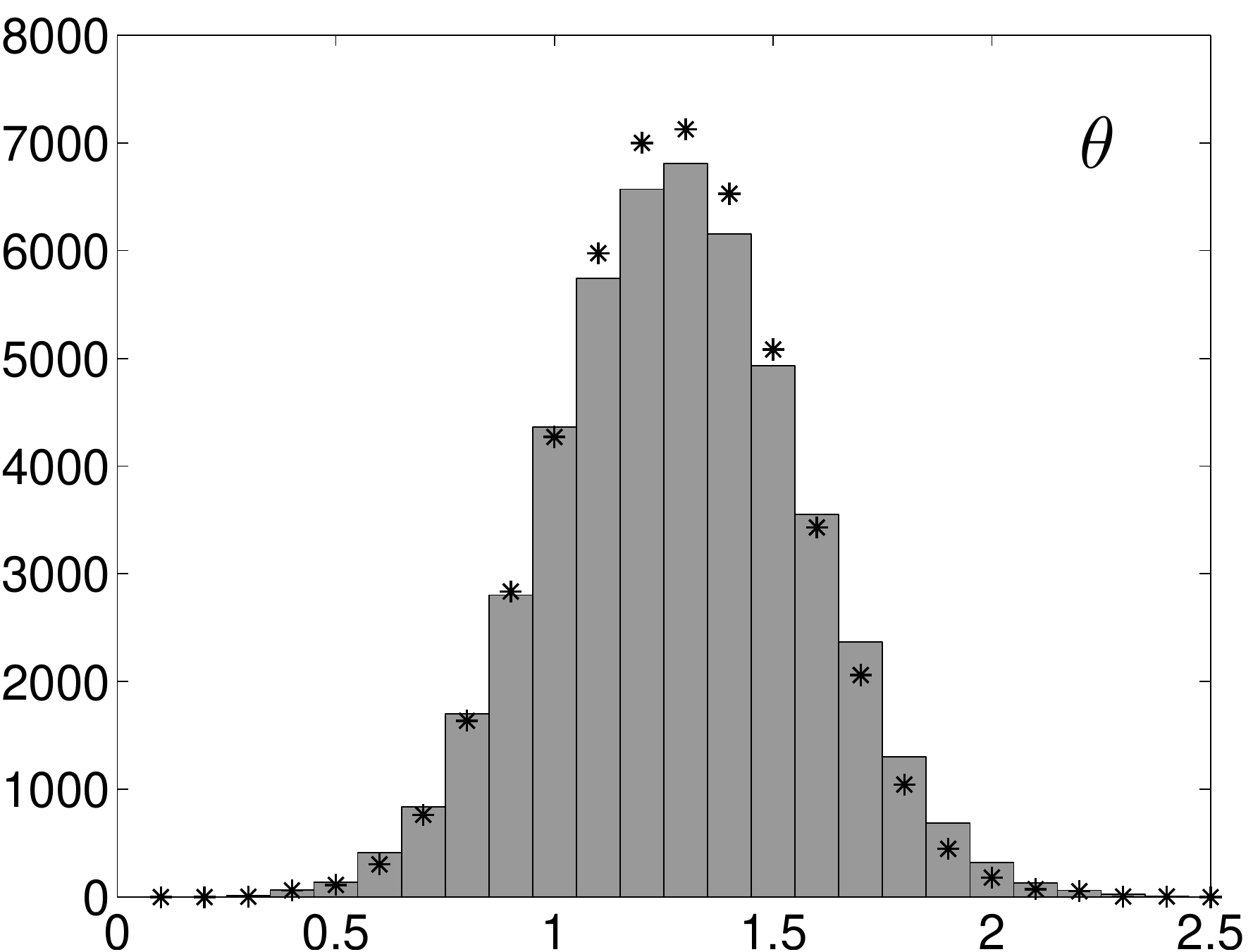}
  \caption{Histograms representing the posterior density $p(\parameter \mid y_{1:\T})$ for the \pgas sampler (gray bars) and for PMMH ($*$).
    The ``true'' value, used in the data generation, is 1.}
  \label{fig:degen1_theta}
\end{figure}

\section{Discussion}\label{sec:discussion}%
\pgas is a novel approach to PMCMC that makes use of backward simulation
ideas without needing an explicit backward pass.  Compared to \pgbs, a
conceptually similar method that does require an explicit backward pass,
\pgas has advantages, most notably for inference in the non-Markovian \ssm{s}
that have been our focus here.  When using the proposed truncation of the 
backward weights, we have found \pgas to be more robust to the approximation 
error than \pgbs.  Furthermore, for non-Markovian models, \pgas is easier 
to implement than \pgbs, since it requires less bookkeeping. 
It can also be more memory efficient, since it does not require us to 
store intermediate quantities that are needed for a separate backward simulation pass,
as is done in \pgbs. Finally, we note that \pgas can be used as an alternative 
to \pgbs for other inference problems to which PMCMC can be applied, and we
believe that it will prove attractive in problems beyond the non-Markovian
\ssm{s} that we have discussed here.

\appendix
\section{Proof of \Prop{truncation_kld}}\label{app:proofs}%
With $M = \T - t$ and $w(k) = w_t^k$, the distributions of interest are given by
\begin{align*}
  P(k) &= \frac{w(k)\prod_{s = 1}^M h_s(k)} {\normsum{l} w(l)\prod_{s = 1}^M h_s(l)} & &\text{ and } &
  \widetilde P_p(k) &= \frac{w(k)\prod_{s = 1}^p h_s(k)} {\normsum{l} w(l)\prod_{s = 1}^p h_s(l)},
\end{align*}
respectively. Let
$\varepsilon_s \triangleq \max_{k,l} \left( h_s(k) / h_s(l)-1 \right) \leq A\exp(-cs)$
and consider
\begin{align*}
  \left( \sum_l w(l) \prod_{s=1}^p h_s(l) \right) \prod_{s = p+1}^M h_s(k) &\leq  \sum_l  w(l) \prod_{s=1}^p h_s(l) \prod_{s = p+1}^M h_s(l)(1+\varepsilon_s) \\
  &= \left(  \sum_l  w(l) \prod_{s=1}^M h_s(l)  \right) \prod_{s = p+1}^M (1+\varepsilon_s) .
\end{align*}
It follows that the KLD is bounded according to,
\begin{align}
  \nonumber
  \KLD(P \| \widetilde P_p) &= \sum_k P(k) \log \frac{P(k)}{\widetilde P_p(k)} \\
  \nonumber
  &= \sum_k P(k) \log \left( \frac{ \prod_{s = p+1}^M h_s(k) \left( \sum_l w(l) \prod_{s=1}^p h_s(l) \right) }{  \sum_l w(l) \prod_{s=1}^M h_s(l) } \right) \\
  \nonumber
  &\leq \sum_k P(k) \sum_{s = p+1}^M \log(1+\varepsilon_s) \leq \sum_{s = p+1}^M \varepsilon_s \leq A \sum_{s = p+1}^M \exp(-cs) \\
  &= A \frac{e^{-c(p+1)} - e^{-c(M+1)}}{1-e^{-c}}.
  \tag*{$\square$}
\end{align}

\bibliographystyle{IEEEtran}
\bibliography{references}

\begin{thebibliography}{10}
\providecommand{\url}[1]{#1}
\csname url@samestyle\endcsname
\providecommand{\newblock}{\relax}
\providecommand{\bibinfo}[2]{#2}
\providecommand{\BIBentrySTDinterwordspacing}{\spaceskip=0pt\relax}
\providecommand{\BIBentryALTinterwordstretchfactor}{4}
\providecommand{\BIBentryALTinterwordspacing}{\spaceskip=\fontdimen2\font plus
\BIBentryALTinterwordstretchfactor\fontdimen3\font minus
  \fontdimen4\font\relax}
\providecommand{\BIBforeignlanguage}[2]{{%
\expandafter\ifx\csname l@#1\endcsname\relax
\typeout{** WARNING: IEEEtran.bst: No hyphenation pattern has been}%
\typeout{** loaded for the language `#1'. Using the pattern for}%
\typeout{** the default language instead.}%
\else
\language=\csname l@#1\endcsname
\fi
#2}}
\providecommand{\BIBdecl}{\relax}
\BIBdecl

\bibitem{LindstenJS:2012}
F.~Lindsten, M.~I. Jordan, and T.~B. Sch\"on, ``Ancestor sampling for particle
  {G}ibbs,'' in \emph{Proceedings of the 2012 Conference on Neural Information
  Processing Systems ({NIPS})}, Lake Tahoe, NV, USA, Dec. 2012.

\bibitem{AndrieuDH:2010}
C.~Andrieu, A.~Doucet, and R.~Holenstein, ``Particle {M}arkov chain {M}onte
  {C}arlo methods,'' \emph{Journal of the Royal Statistical Society: Series B},
  vol.~72, no.~3, pp. 269--342, 2010.

\bibitem{WhiteleyAD:2010}
N.~Whiteley, C.~Andrieu, and A.~Doucet, ``Efficient {B}ayesian inference for
  switching state-space models using discrete particle {M}arkov chain {M}onte
  {C}arlo methods,'' Bristol Statistics Research Report 10:04, Tech. Rep.,
  2010.

\bibitem{LindstenS:2012}
F.~Lindsten and T.~B. Sch\"{o}n, ``On the use of backward simulation in the
  particle {G}ibbs sampler,'' in \emph{Proceedings of the 2012 {IEEE}
  International Conference on Acoustics, Speech and Signal Processing
  ({ICASSP})}, Kyoto, Japan, Mar. 2012.

\bibitem{DoucetJ:2011}
A.~Doucet and A.~Johansen, ``A tutorial on particle filtering and smoothing:
  Fifteen years later,'' in \emph{The Oxford Handbook of Nonlinear Filtering},
  D.~Crisan and B.~Rozovsky, Eds.\hskip 1em plus 0.5em minus 0.4em\relax Oxford
  University Press, 2011.

\bibitem{PittS:1999}
M.~K. Pitt and N.~Shephard, ``Filtering via simulation: Auxiliary particle
  filters,'' \emph{Journal of the American Statistical Association}, vol.~94,
  no. 446, pp. 590--599, 1999.

\bibitem{DykP:2008}
D.~A.~V. Dyk and T.~Park, ``Partially collapsed {G}ibbs samplers: Theory and
  methods,'' \emph{Journal of the American Statistical Association}, vol. 103,
  no. 482, pp. 790--796, 2008.

\bibitem{Whiteley:2010}
N.~Whiteley, ``Discussion on {P}article {M}arkov chain {M}onte {C}arlo
  methods,'' Journal of the Royal Statistical Society: Series {B}, 72(3), p
  306--307, 2010.

\bibitem{ChenL:2000}
R.~Chen and J.~S. Liu, ``Mixture {K}alman filters,'' \emph{Journal of the Royal
  Statistical Society: Series {B}}, vol.~62, no.~3, pp. 493--508, 2000.

\bibitem{DoucetGA:2000}
A.~Doucet, S.~J. Godsill, and C.~Andrieu, ``On sequential {M}onte {C}arlo
  sampling methods for {B}ayesian filtering,'' \emph{Statistics and Computing},
  vol.~10, no.~3, pp. 197--208, 2000.

\bibitem{SchonGN:2005}
T.~Sch\"{o}n, F.~Gustafsson, and P.-J. Nordlund, ``Marginalized particle
  filters for mixed linear/nonlinear state-space models,'' \emph{{IEEE}
  Transactions on Signal Processing}, vol.~53, no.~7, pp. 2279--2289, Jul.
  2005.

\bibitem{SarkkaBG:2012}
S.~S{\"a}rkk{\"a}, P.~Bunch, and S.~Godsill, ``A backward-simulation based
  {R}ao-{B}lackwellized particle smoother for conditionally linear {G}aussian
  models,'' in \emph{Proceedings of the 16th {IFAC} Symposium on System
  Identification}, Brussels, Belgium, Jul. 2012.

\bibitem{FongGDW:2002}
W.~Fong, S.~J. Godsill, A.~Doucet, and M.~West, ``Monte {C}arlo smoothing with
  application to audio signal enhancement,'' \emph{{IEEE} Transactions on
  Signal Processing}, vol.~50, no.~2, pp. 438--449, Feb. 2002.

\bibitem{HjortHMW:2010}
N.~L. Hjort, C.~Holmes, P.~Müller, and S.~G. Walker, Eds., \emph{Bayesian
  Nonparametrics}.\hskip 1em plus 0.5em minus 0.4em\relax Cambridge University
  Press, 2010.

\bibitem{MacEachernCL:1999}
S.~N. MacEachern, M.~Clyde, and J.~S. Liu, ``Sequential importance sampling for
  nonparametric {B}ayes models: The next generation,'' \emph{The Canadian
  Journal of Statistics}, vol.~27, no.~2, pp. 251--267, 1999.

\bibitem{Fearnhead:2004}
P.~Fearnhead, ``Particle filters for mixture models with an unknown number of
  components,'' \emph{Statistics and Computing}, vol.~14, pp. 11--21, 2004.

\bibitem{RasmussenW:2006}
C.~E. Rasmussen and C.~K.~I. Williams, \emph{Gaussian Processes for Machine
  Learning}.\hskip 1em plus 0.5em minus 0.4em\relax {MIT} Press, 2006.

\bibitem{Bouchard-CoteSJ:2012}
A.~Bouchard-C\^ot\'e, S.~Sankararaman, and M.~I. Jordan, ``Phylogenetic
  inference via sequential {M}onte {C}arlo,'' \emph{Systematic Biology},
  vol.~61, no.~4, pp. 579--593, 2012.

\bibitem{TehDR:2008}
Y.~W. Teh, H.~{Daum\'e III}, and D.~Roy, ``Bayesian agglomerative clustering
  with coalescents,'' \emph{Advances in Neural Information Processing}, pp.
  1473--1480, 2008.

\bibitem{Bierman:1973}
G.~J. Bierman, ``Fixed interval smoothing with discrete measurements,''
  \emph{International Journal of Control}, vol.~18, no.~1, pp. 65--75, 1973.

\bibitem{GodsillDW:2004}
S.~J. Godsill, A.~Doucet, and M.~West, ``{M}onte {C}arlo smoothing for
  nonlinear time series,'' \emph{Journal of the American Statistical
  Association}, vol.~99, no. 465, pp. 156--168, Mar. 2004.

\bibitem{ArulampalamMGC:2002}
M.~S. Arulampalam, S.~Maskell, N.~Gordon, and T.~Clapp, ``A tutorial on
  particle filters for online nonlinear/non-{G}aussian {B}ayesian tracking,''
  \emph{{IEEE} Transactions on Signal Processing}, vol.~50, no.~2, pp.
  174--188, 2002.

\bibitem{ArulampalamRGM:2004}
M.~S. Arulampalam, B.~Ristic, N.~Gordon, and T.~Mansell, ``Bearings-only
  tracking of manoeuvring targets using particle filters,'' \emph{EURASIP
  Journal on Applied Signal Processing}, vol.~15, pp. 2351--2365, 2004.

\bibitem{RisticAG:2004}
B.~Ristic, S.~Arulampalam, and N.~Gordon, \emph{Beyond the {K}alman filter:
  particle filters for tracking applications}.\hskip 1em plus 0.5em minus
  0.4em\relax London, UK: Artech House, 2004.

\end{thebibliography}

\end{document}